\documentclass[conference,letterpaper]{IEEEtran}

\usepackage{style}
\usepackage[utf8]{inputenc} 
\usepackage[T1]{fontenc}
\usepackage{url}
\usepackage{ifthen}
\usepackage{cite}
\usepackage[cmex10]{amsmath} 
\usepackage{mathrsfs}
\usepackage{amssymb,bm}
\usepackage{mathtools}
\usepackage{xcolor}
\usepackage{tikz}
\usepackage{pgfplots}
\usepackage{subcaption}
\usepackage{microtype}
\usepackage{enumitem}
\usepackage{hyperref}

  {\par\medskip\noindent\textcolor{magenta}{\rule{\linewidth}{0.4pt}}\\%
   \textcolor{magenta}{\bfseries EDIT NOTE:}\ \textcolor{magenta}\bgroup}%
  {\egroup\par\noindent\textcolor{magenta}{\rule{\linewidth}{0.4pt}}\medskip}
  {\par\medskip\noindent\textcolor{teal}{\rule{\linewidth}{0.6pt}}\\%
   \textcolor{teal}{\bfseries SUGGESTED TEXT #1}\par\smallskip\noindent}%
  {\par\noindent\textcolor{teal}{\rule{\linewidth}{0.6pt}}\medskip}

\newif\iflongpaper

\longpapertrue

\begin{document}
\title{Lossy Compression via Sparse Regression Codes: Generalized Construction and Finite-length Bounds} 


\author{%
  \IEEEauthorblockN{Galen Reeves}
  \IEEEauthorblockA{ 
                    Duke University\\
                    Email: galen.reeves@duke.edu}
  \and
  \IEEEauthorblockN{Ramji Venkataramanan}
  \IEEEauthorblockA{ 
    University of Cambridge\\
    Email: rv285@cam.ac.uk}
}

\maketitle

\begin{abstract}
We study sparse regression codes (SPARCs) for lossy compression under simple greedy encoding rules, including both correlation-based and distance-based methods. We generalize the SPARC construction, and consider the class of \emph{additive orthogonal} regression codes, of which standard SPARCs are a special case.  For this class of codes, we derive nonasymptotic bounds on the squared-error distortion by tracking the evolution of the encoding residual across stages. Our results highlight the role of power allocation in controlling the distortion, allowing us to optimize the allocation based  on the  parameters of the code. The optimized allocation improves the finite-length compression performance of SPARCs, and our bounds provide distortion guarantees for lower complexity variants of SPARCs, like signed SPARCs and $K$-sparse SPARCs.
\end{abstract}

\section{Introduction} \label{sec:intro}

We consider efficient lossy compression under the squared error distortion criterion using sparse regression codes (SPARCs). 
SPARCs were proposed by Barron and Joseph \cite{JosephBarron_ML12} for communication over  additive white Gaussian noise channels, and have been proven to asymptotically achieve the channel capacity with various efficient decoders \cite{JosephBarron_Fast12, BarronC12,BarbKrz17,RushGV17}. For lossy compression, SPARCs have been shown to asymptotically achieve the optimal distortion-rate tradeoff for an i.i.d. Gaussian source with a simple successive cancellation encoder \cite{kontoyiannis2010sparse, RVGaussFeasible}. 
Moreover, with optimal (minimum-distance) encoding, they also achieve the optimal excess-distortion exponent \cite{RVsparcRD_ML18}. Despite these attractive asymptotic properties, at finite block-lengths there is a significant gap between the empirical distortion achieved by SPARCs with successive cancellation encoding  and the optimal distortion-rate tradeoff \cite{RVGaussFeasible}. We recall that for an i.i.d. Gaussian source with variance $\sigma^2$, the optimal distortion-rate function at rate $R$ nats/sample is $D^*(R) = \sigma^{2} e^{-2R}$ \cite{coverT12}.  Wu et al. \cite{wu2023_decimation} proposed an iterative soft-decision encoder for SPARCs based on Approximate Message Passing (AMP) with improved empirical performance at finite lengths (compared to successive cancellation),  but there are no theoretical guarantees on the achievable distortion.

Before outlining our contributions, we describe the general structure of a regression code, of which the sparse regression code is a special case. A regression code has a codebook of the form 
\begin{align}
     \{ A \beta \in \bbR^n \mid \beta \in \cB \},
     \label{eq:regcode_def}
\end{align}
where $A \in \reals^{n \times N}$ is a design matrix (typically chosen to have i.i.d. Gaussian entries), and $\cB \subset \reals^N$ is a set of structured coefficient vectors. Here $n$ is the code length, and the rate of the code is $\frac{\log \abs{\cB}}{n}$, where $\abs{\cB}$ denotes the cardinality of $\cB$. Unless otherwise mentioned, we use natural logarithms and measure rate in nats.

The representation in \eqref{eq:regcode_def} separates the roles of the code components: the structure of $\cB$ enables computationally efficient encoding and decoding, while the design matrix  $A$ imparts favorable geometric properties to the codebook.  In particular, when $A$ has i.i.d.\ Gaussian entries and $\cB$ is a suitably chosen set of sparse vectors, the resulting codewords are approximately isotropic, and the codebook behaves like a random $\ell_2$-covering of $\reals^n$. As a result, sparse regression codes are well-suited for both Gaussian source coding and channel coding problems. In the context of channel coding, the partitioned structure of the SPARC codebook facilitates combining them with standard outer codes \cite{greig2017techniques,EbertSPARCLDPC}, and also makes it adaptable to multi-user Gaussian channels \cite{venkataramanan19monograph, fengler2021SPARCS,amalladinne2022unsourced, hsieh2022near,yan2026capacityregionachievingSPARC}.

\begin{figure}
    \centering
    \includegraphics[width=0.98\linewidth]{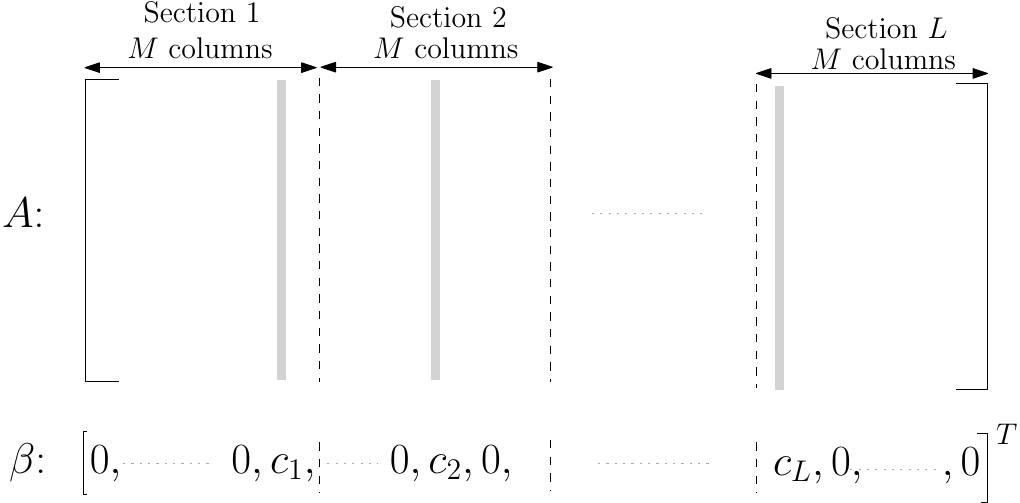}
    \caption{Structure of a standard SPARC. Codewords are of the form $A\beta$, where $\beta \in \reals^{ML}$ has exactly one nonzero component in each of its $L$ sections. }
    \label{fig:SPARC_fig}
\end{figure}

In the standard sparse regression construction \cite{JosephBarron_ML12,JosephBarron_Fast12, RVGaussFeasible}, the set $\cB$ consists of sparse vectors $\beta$, each with $L$ sections and exactly one non-zero per section. As shown in Fig. \ref{fig:SPARC_fig}, we write $N=ML$, so that each section has $M$ columns. The value of the nonzero component in section $l \in [L]$ is denoted $c_l$. These values  $\{ c_l \}_{l\in [L]}$ are fixed beforehand and known to both the encoder and the decoder. Since there are $M$ choices for the location of the nonzero in each of the $L$ sections, the number of codewords is $M^L$ and the rate is $\frac{L \log M}{n}$.

In this paper, we revisit sparse regression codes with successive cancellation (greedy) encoding. We study a general class of \emph{additive orthogonal} regression codes, of which SPARC is a special case. We characterize the performance under both correlation-based and distance-based greedy encoding rules, via deterministic predictions and nonasymptotic bounds on the distortion achieved with each of these methods. These predictions allow us to optimize the non-zero coefficients defining the code based on the finite-length parameters of the code. \iflongpaper (Borrowing terminology from channel coding, we call the choice of non-zero coefficients the ``power allocation".)\fi Our results show how the original SPARC construction can be generalized to improve finite-length performance, while maintaining the simplicity of greedy encoding.

\iflongpaper
The rest of the paper is organized as follows. Section~\ref{sec:additive} introduces additive orthogonal regression codes and the two greedy encoding rules. Section~\ref{sec:main_results} presents our main concentration and distortion-prediction results, together with the optimized power allocations for both encoding rules. Section~\ref{sec:experiments} reports numerical experiments validating the predictions, and Section~\ref{sec:discussion} discusses the implications and directions for future work.
\fi

\section{Additive Orthogonal Regression Codes} \label{sec:additive}
\iflongpaper
We now introduce a class of regression codes that is broad enough to include signed and $K$-sparse SPARCs alongside the standard construction, but structured enough to admit a common stage-wise encoder and a unified distortion analysis.
\fi

 We say that a regression code~\eqref{eq:regcode_def}  is additive if its  set of coefficient vectors $\cB$ admits a decomposition of the form
\begin{align}
\cB = \cB_1 + \cdots + \cB_L 
\coloneqq \left\{ \sum_{l=1}^L \beta_l \;\middle|\; \beta_l \in \cB_l \right\}, 
\quad \cB_l \subset \bbR^N. 
\label{eq:cB_additive} 
\end{align}
 Furthermore, in an \emph{additive orthogonal} regression code, the sets $\{\cB_l\}_{l=1}^L$ are mutually orthogonal (i.e., $\langle b, b' \rangle = 0$ for $b \in \cB_l$, $b' \in \cB_{l'}$ with $l \neq l'$), and each vector $b \in \cB_l$ satisfies $\|b \| =
 c_l$, for a constant $c_l >0$. This structure induces a natural stage-wise representation of codewords and leads to greedy decoding procedures that construct the coefficient vector incrementally, one component at a time. The rate of an additive orthogonal regression code is $\frac{1}{n} \sum_{l=1}^L \log \abs{\cB_l}$.  

SPARCs (Fig. \ref{fig:SPARC_fig}) are additive orthogonal  regression codes, with $N=ML$ and where each $\cB_l$ is a block of one-sparse vectors of the form 
\begin{align*}
\cB_l = \{ c_l e_j  \mid  (l-1) M < j   \le l M \},
\end{align*}
where $e_j$ is the canonical unit vector with $1$ in position $j$. Other examples of additive orthogonal regression codes include: 
\begin{itemize}[leftmargin=1.2em,itemsep=1pt, topsep=2pt, parsep=0pt]
    \item \emph{Signed} SPARCs \cite{hsieh2021modulated}, where the non-zero coefficient in section $l$ can be either $c_l$ or $-c_l$,
    \item \emph{$K$-sparse} SPARCs, which have $K$ nonzero coefficients in each section $l$, each with value $c_l/ \sqrt{K}$ (the standard SPARC  corresponds to $K=1$), 
\end{itemize}

With $L$ sections and $M$ columns per section, the rate of the standard SPARC is $\frac{L \log (M)}{n}$. The  signed SPARC  has rate $\frac{L \log (2M)}{n}$, and the $K$-sparse SPARC $\frac{L \log {M \choose K}}{n}$. Equivalently, for a fixed rate $R$, these SPARC variants enable smaller design matrices, and lower storage and encoding complexity.

\subsection{Greedy Encoding Rules} \label{subsec:greedy_rules}

Given a source vector $x \in \bbR^n$, the optimal encoder for a regression code determines a coefficient vector $\beta \in \mathcal{B}$ that minimizes the reconstruction error, i.e.,
\begin{align}
\min_{\beta \in \mathcal{B}} \| x - A \beta \|. \label{eq:decoding_obj}
\end{align}
Solving \eqref{eq:decoding_obj} exactly is generally computationally intractable. A simple practical alternative is to  construct 
$\beta$  incrementally using greedy or stage-wise procedures.

\paragraph*{Distance-based greedy encoding} 
A natural approach is to construct the coefficient vector sequentially by performing a minimum-distance update at each stage:\begin{subequations}
\label{eq:inc_residual}
\begin{align}
b_l &= \argmin_{b \in \cB_l} \| r_{l-1} - A b \|, \\
\beta_l &= \beta_{l-1} + b_l, \\
r_l &= x - A \beta_l,
\end{align}
\end{subequations}
where $\beta_0 = 0$ and $r_0 = x$. 
At each stage, $\beta_l$ can be viewed as an approximate solution to~\eqref{eq:decoding_obj} restricted to the partial coefficient set $\cB_1 + \cdots + \cB_l$.

\paragraph*{Correlation-based greedy encoding} 
An alternative is 
to use a correlation-based selection rule: \vspace{-.1cm}
\begin{subequations}
\label{eq:max_ip}
\begin{align}
b_l &= \argmax_{b \in \cB_l} \langle r_{l-1}, A b \rangle, \label{eq:bl_update_corr} \\
\beta_l &= \beta_{l-1} + b_l, \\
r_l &= x - A \beta_l.
\end{align}
\end{subequations}
Although the correlation-based  rule is generally suboptimal compared to the distance-based one~\eqref{eq:inc_residual}, it is computationally lighter and more amenable to analysis. For sparse regression codes with i.i.d.~Gaussian design matrices, correlation-based  encoding achieves the Gaussian rate-distortion function \cite{RVGaussFeasible}. 

The correlation-based encoder  has  a simple form for both $K$-sparse SPARCs and signed SPARCs. For $K$-sparse SPARCs, the maximization in \eqref{eq:bl_update_corr} is equivalent to choosing the $K$ largest values among  $  \langle A^\top r_{l-1}, e_j\rangle $ for $(l-1)M < j \le lM$. For signed SPARCs, the update rule is equivalent to choosing the largest among   $  \pm \langle A^\top r_{l-1}, e_j\rangle $ for $(l-1)M < j \le lM$. 

\section{Main Results} \label{sec:main_results}

We analyze both correlation-based and distance-based greedy encoding rules for additive orthogonal regression codes. Our analysis provides a deterministic prediction for the residual error across decoding stages and shows that the achieved distortion concentrates sharply around this prediction. 
\iflongpaper
All proofs are deferred to Appendix \ref{app:proofs_main_results}.
\fi

\begin{assumption}\label{ass:source_coding}
We assume throughout the paper that the design matrix $A \in \bbR^{n \times N}$ has i.i.d.\ $\normal(0, 1)$ entries, and is independent of the source vector $x \in \bbR^n$ (which can be deterministic or arbitrarily distributed).
\end{assumption}

\subsection{Concentration}

\begin{theorem}[Concentration]\label{thm:concentration}
Consider an additive orthogonal regression code with either distance-based or correlation-based encoding. Then, for each stage $l \in [L]$ there exists  a $1$-Lipschitz function $\mu_l  \colon \bbR_+ \to \bbR_+$ such that:
\begin{enumerate} 
\item  
The residual $r_l =x - A \beta_l$ at each stage satisfies
\begin{align*}
\E[ \|r_l\| \mid x ] = \mu_l(\|x\|) , \quad \forall x \in \bbR^n
\end{align*}
\item The centered  norm $\|r_l\| - \mu_l(\|x\|)$ is sub-Gaussian with variance proxy $\sum_{k=1}^l c_k^2$.  In particular, for all $t \ge 0$, 
    \begin{align*}
        \P\Big( \big| \|r_l\| - \mu_l(\|x\|) \big| \ge  t\Big) \le 2 \exp\left\{ - \frac{ t^2}{2 \sum_{k=1}^l c_k^2} \right\}.
    \end{align*}
\end{enumerate}
\end{theorem}

Theorem~\ref{thm:concentration} shows that the residual norm concentrates at a rate controlled by the total power $\sum_{k=1}^L c_k^2$. This concentration holds uniformly over all realizations of the source, and the centering term $\mu_l(\|x\|)$ depends only on the source magnitude.

For a zero-mean ergodic source with variance $\sigma^2$, we have 
\begin{align*}
    \frac{1}{\sqrt{n}} \|x\| \xrightarrow[n \to \infty]{} \sigma
\end{align*}
both in probability and in mean square. In this setting, it is natural to design the power allocation based on $\sigma^2$ and to analyze performance in terms of this parameter. In the following, we show that the optimal power allocation (and more generally, any near-optimal allocation) satisfies
\begin{align}
\sum_{k=1}^L c_k^2 \le \sigma^2, \label{eq:ckUB}
\end{align}
so that the concentration level does not scale with  $n$. 

Moreover, since $\mu_l(\cdot)$ is $1$-Lipschitz, we may consider  a deterministic centering given by $\mu_l(\sqrt{n}\sigma)$. To  align with standard performance metrics, we define:
\begin{align}
D_l \coloneqq \frac{1}{n} \mu_l^2(\sqrt{n} \sigma), \label{eq:Dl}
\end{align}
which corresponds to the limiting mean squared error (i.e., the distortion). 

\begin{cor}
For every $\eps > 0$ and $l \in [L]$, the event
\begin{align*}
\left| \frac{1}{\sqrt{n}} \|r_l\| - \sqrt{D_l} \right|
\le 
\left| \frac{1}{\sqrt{n}}\|x\| - \sigma \right| + \epsilon
\end{align*}
holds with  probability at least $1-2 \exp\{ - n \eps^2 / (2 \sum_{k=1}^l c_k^2)\}$. 
\end{cor}

In the following, we provide upper and lower bounds on $D_l$ as a function of the power allocation $\{c_l\}_{l=1}^L$ and structural properties of the code. This is achieved via a tractable approximation to the centering $D_l$ that we call the `deterministic prediction' $\Delta_l$. The $\{ \Delta_l \}$ yield explicit recursions that can be used to optimize the power allocation. For correlation-based encoding, we show the  difference $\abs{D_l - \Delta_l^{\mathrm{corr}}}$ is negligible under this power allocation  (Theorem~\ref{thm:D_corr_opt}). Such bounds are harder to obtain for distance-based encoding, but we show that for standard SPARCs, $\Delta_l^{\mathrm{dist}}$ is no larger than $\Delta_l^{\mathrm{corr}}$.

\subsection{Correlation-based encoding} 

\begin{definition}
For each stage $l \in [L]$, define the normalized Gaussian width
\begin{align}
\label{eq:omega_l_def}
\omega_l  \coloneqq \E \left[  \max_{b \in \cB_l} \frac{  \langle z, b \rangle}{\|b\| } \right], \qquad z \sim \normal(0, \Id_N) 
\end{align} 
This quantity depends only on the directions in $\mathcal{B}_l$ and satisfies $
0 \le \omega_l \le \sqrt{2 \log |\mathcal{B}_l|}$ \cite{boucheron2013concentration}.
\end{definition}

\begin{theorem}\label{thm:D_corr}
Consider an additive  orthogonal regression code with parameters $\{(c_l, \omega_l)\}_{l=1}^L$. Under correlation-based encoding, the deterministic centering  in \eqref{eq:Dl} satisfies a recursion of the form 
\begin{align*}
D_0= \sigma^2, \quad 
D_l =   \Big( \sqrt{ D_{l-1}  } -   \frac{ \omega_l c_l }{ \sqrt{n}}  \Big)^2 + c_l^2 + \frac{\eta_{l}}{n}, 
\end{align*}
where $|\eta_l| \le 2\sum_{k=1}^l c_k^2$. 
\end{theorem}

For the purposes of analysis, we consider an approximation to this  recursion that omits the $\eta_l$ terms. Specifically, we define  
\begin{align}
\Delta_0 =\sigma^2 , \quad 
\Delta_l =   \Big( \sqrt{ \Delta _{l-1}  } -   \frac{ \omega_l c_l }{ \sqrt{n}}  \Big)^2+    c^2_l. \label{eq:Delta_l_corr} 
\end{align}
The numerical simulations in Figure \ref{fig:corr_based} show a close match with the theoretical prediction $\Delta_L$, for different power allocations and SPARC variants.

Since each update in \eqref{eq:Delta_l_corr} is quadratic in $c_l$, minimizing sequentially over the power allocation yields
\begin{align}
c^2_l = \frac{ \sigma^2 \omega^2_l}{n + \omega_l^2}  \prod_{k=1}^{l} \frac{n}{n+ \omega_k^2}
\label{eq:c_corr}.
\end{align}
This allocation satisfies the total power constraint in \eqref{eq:ckUB} and results in
\begin{align}
\Delta^\mathrm{corr}_l & \coloneqq \sigma^2 \prod_{k=1}^l \frac{n}{n+\omega_k^2}. \label{eq:Delta_corr_opt} 
\end{align}
\iflongpaper
The derivation of the optimal power allocation, including verification of the power constraint and the resulting value of $\Delta^\mathrm{corr}_l$, 
is provided in Appendix~\ref{app:power_alloc_corr}.
\fi

The next theorem shows that this power allocation is near-optimal for the deterministic centering $\{D_l\}$. This result does not follow directly from Theorem~\ref{thm:D_corr}, but instead requires a refined analysis that exploits the specific structure of \eqref{eq:c_corr}.

\begin{theorem}\label{thm:D_corr_opt} Consider an additive  orthogonal regression code with correlation-based encoding and $n > 2$. Under the power allocation \eqref{eq:c_corr}, the deterministic centering in \eqref{eq:Dl} satisfies 
\begin{align*}
D_l \le \Delta^\mathrm{corr}_l +  2 ( \sigma^2 - \Delta^\mathrm{corr}_l) \eps_{l-1}.
\end{align*}
where
\begin{align*}
\eps_{l} 
& \coloneqq  1-  \prod_{k=1}^{l} \frac{ (n-2 ) (n+ \omega_k^2)}{n (n-2+\omega_k^2)} \le  \frac{2}{n} \sum_{k=1}^{l} \frac{ \omega_k^2}{  n - 2 + \omega_k^2}.
\end{align*}
Conversely, for any power allocation, we have the lower bound 
\begin{align*}
D_l & \ge \Delta_l^\mathrm{corr} (1- \eps_l).
\end{align*}
\end{theorem}

\begin{remark} A nonasymptotic bound on the distortion  of SPARCs under correlation-based encoding  was obtained  in \cite[Theorem 1]{RVGaussFeasible}. Our results improve upon the result in  \cite{RVGaussFeasible} in the sense that: i) they hold for any additive orthogonal regression code,  ii) concentration is proved under any choice of power allocation, and iii) the near-optimality of the power allocation in \eqref{eq:c_corr} is established by tight upper and lower bounds on the deterministic centering in Theorem~\ref{thm:D_corr_opt}. 
\end{remark}

Under the power allocation \eqref{eq:c_corr},  for a fixed $n,L$, the predicted distortion $\Delta^\mathrm{corr}_L$ is a function of the Gaussian widths $\{ \omega_l \}$, which are determined by the coefficient sets $\{ \cB_l \}$ defining the code. The following proposition specifies the Gaussian widths for the SPARC,  signed SPARC, and $K$-SPARC.

\begin{prop} \label{prop:Gaussian_width}
Consider a $K$-sparse SPARC (defined in Section \ref{sec:additive}) with $L$ sections and $M$ columns per section,  and let $\omega_{l,K}$ denote its Gaussian width \eqref{eq:omega_l_def}, for $l \in [L]$. (Recall that $K=1$ corresponds to the standard SPARC.)  Furthermore, let $\omega_{l, \pm}$  denote the Gaussian width of the signed SPARC with the same values of $L,M$. Then, letting $Z_1, \ldots, Z_M$ $\sim_{\text{i.i.d.}}  \normal(0,1)$, we have for $l \in [L]$: 
\begin{itemize}[leftmargin=1.2em,itemsep=1pt, topsep=2pt, parsep=0pt]
        \item Standard SPARC:   $\omega_{l,1} = \E\left[ \max_{1 \le j \le M} \, Z_j \right]$,
        \item Signed SPARC: $\omega_{l, \pm} = \E\left[ \max_{1 \le j \le M} \, \abs{Z_j} \right]$,
        \item $K$-sparse SPARC: 
        \[ \omega_{l, K} = 
        \E[ Z_{(M)} + Z_{(M-1)} + \ldots + Z_{(M-K+1)} ]/\sqrt{K}, \] where $Z_{(M)}$ denotes the maximum  of $Z_1, \ldots, Z_M$ $\sim_{\text{i.i.d.}}  \normal(0,1)$, and $Z_{(M-1)}$ denotes the second maximum (second-largest value), and so on. 
\end{itemize}
Consequently, for fixed $K$, these Gaussian widths admit the following
asymptotic expansions as $M \to \infty$. Define
\begin{align*}
\psi(M) \coloneqq \sqrt{2 \log M} - \frac{\log(4\pi \log M) - 2\gamma}{2\sqrt{2 \log M}},
\end{align*}
where $\gamma \approx 0.5772$ is the Euler-Mascheroni constant and let $H_K = \sum_{j=1}^K \frac{1}{j}$ be the $K$th Harmonic number. Then, we have
\begin{align}
\omega_{l,1} &= \psi(M) + \mathcal{O}\!\left(\rho_M \right), \label{eq:w_sparc_asymp} \\
\omega_{l,\pm} &= \psi(2M) + \mathcal{O}\!\left(\rho_M\right), \label{eq:w_sparc_pm_asymp} \\
\omega_{l,K} &= \sqrt{K}\,\psi(M) - \frac{K H_{K-1} - K + 1}{\sqrt{2K \log M}} + \mathcal{O}\!\left(\rho_M \right), \label{eq:w_2sparc_asymp}
\end{align}
where  $\rho_M \coloneqq (\log \log M)^2 (\log M)^{-3/2}$. 
\end{prop}

\begin{remark}
The analysis in \cite{RVGaussFeasible} used power allocation:
\begin{align}
c_l^2 = \frac{ \sigma^2  2 \log M }{ n} \left( 1 - \frac{ 2 \log M}{ n} \right)^{l-1}.
\label{eq:asymp_PA}
\end{align}
If we use the heuristic  $\omega_l = \sqrt{2\log M}$ then the power allocation in \eqref{eq:c_corr} for an $(n,L,M)$ SPARC satisfies 
\begin{align}
    c_l^2 & =\frac{ \sigma^2  2 \log M}{ n + 2 \log M} \left(  1-  \frac{ 2 \log M}{n + 2 \log M} \right)^{l}. 
    \label{eq:cl_heuristic}
\end{align}
Thus, we expect  the power allocation in \cite{RVGaussFeasible} to be near optimal when $n \gg \log M$ (or equivalently $L \gg R$), but suboptimal otherwise; see Figure \ref{subfig:exp1} and discussion in Section \ref{sec:experiments}.
\end{remark}

\subsection{Distance-based encoding}

\begin{definition}
For each block $l$, define  $\gamma_l \colon \bbN \times  \bbR \to \bbR_+$ as
\begin{align}
\gamma_l(n, \lambda)   \coloneqq \E\left[ \min_{b \in \cB_l}  \Big \|  e_1 -  \frac{  \lambda}{ \|b\|}   A b  \Big\| \right]. 
\label{eq:gamma_l_def}
\end{align} 
where $e_1 \in \bbR^n $ is the first standard basis vector  and $A \in \bbR^{n \times N}$ has i.i.d. $\normal(0,1)$ entries. 
\end{definition}

\begin{theorem}\label{thm:D_dist}
Consider an additive orthogonal regression code with parameters $\{(c_l, \gamma_l)\}_{l=1}^L$. Under distance-based encoding, the deterministic centering in \eqref{eq:Dl} satisfies a recursion of the form
\begin{align*}
D_0= \sigma^2, \quad 
\sqrt{D_l} =  \sqrt{D_{l-1} } \, \gamma_l \Big(n,  \frac{c_l }{ \sqrt{ n  D_{l-1} }}\Big )  + \frac{\eta_{l}}{\sqrt{n}} 
\end{align*}
where $|\eta_l| \le (\sum_{k=1}^{l-1} c_k^2)^{1/2}$.
\end{theorem}

Similar to the analysis of correlation-based encoding, we introduce an approximation that  omits the $\eta_l$ terms. Define 
\begin{align}
\Delta_0 =\sigma^2 , \quad 
\Delta_l &= \Delta_{l-1}  \, \gamma^2_l \Big(n,  \frac{c_l }{ \sqrt{ n  \Delta_{l-1} }}\Big ) \label{eq:Delta_l_dist} 
\end{align}
The optimal power allocation with respect to this sequence can be obtained by sequentially minimizing $\Delta_l$ in \eqref{eq:Delta_l_dist}. The resulting distortion prediction is 
\begin{align}
 \Delta^\mathrm{dist}_l & = \sigma^2 \prod_{k=1}^l \bar{ \gamma}^2_{k}(n) 
\end{align}
where $\bar{\gamma}_l(n) \coloneqq \inf_{\lambda \in \bbR^+} \gamma_l(n, \lambda)$. 
Unlike the correlation-based case, this optimal allocation is defined only implicitly. Writing $\lambda^*_l \in \argmin_{\lambda \in \bbR^+} \gamma_l(n,\lambda)$, the minimizing coefficient at stage $l$ is $c_l = \lambda^*_l \sqrt{n \Delta_{l-1}}$.

The following result specifies $\bar{ \gamma}_{l}(n)$ for standard SPARCs, signed SPARCs, and $K$-sparse SPARCs.

\begin{prop}
\label{prop:gamma_expressions}
Consider a $K$-sparse SPARC  with $L$ sections and $M$ columns per section,  and let $\bar{ \gamma}_{l,K}(n)$ denote its distance-based distortion parameter, for $l \in [L]$.  Furthermore, let $\bar{ \gamma}_{l,\pm}(n)$  denote the corresponding parameter for the signed SPARC with the same values of $L,M$. Then,  letting $z_1, \dots, z_M \sim_{\text{i.i.d.}} \,   \normal(0, \Id_n)$, we have for $l \in [L]$:
\begin{itemize}[leftmargin=1.2em,itemsep=1pt, topsep=2pt, parsep=0pt]
   \item Standard SPARC:
    \[
     \bar{ \gamma}_{l,1}(n)  = \inf_{\lambda \in \bbR^+} \,  \E\left[  \min_{ 1 \le j \le M} \| e_1 - \lambda z_j\| \right] 
    \]
        
    \item Signed SPARC:  
    \[ \bar{ \gamma}_{l,\pm}(n)  = \inf_{\lambda \in \bbR^+} \E\left[  \min_{ 1 \le j \le M}  \, \min \left\{ \| e_1 - \lambda z_j\| , \,      \| e_1 + \lambda z_j\| \right\} \right]
     \]
    \item
$K$-sparse SPARC: 
\[\bar{ \gamma}_{l,K}(n)  = \inf_{\lambda \in \bbR^+} \!  \E\left[  \min_{ 1 \le i_1 < \ldots<i_K   \le M} \norm{ e_1 - \frac{\lambda}{ \sqrt{K}} \sum_{k=1}^K z_{i_k} } \right]
\]
\end{itemize}

\end{prop}

For distance-based encoding, it is challenging to obtain tight bounds analogous to Theorem \ref{thm:D_corr_opt}  relating the deterministic centering $D_l$ to $\Delta^\mathrm{dist}_l$. However, the following result shows that for standard SPARCs, $\Delta^\mathrm{dist}_l$  is no larger than  $\Delta^\mathrm{corr}_l$. 
\begin{lemma}\label{lem:gammabar_UB}%
For a standard SPARC, we have
\begin{align}
\bar{ \gamma}^2_{l,1}(n)  &\le   \frac{n}{n+ \omega_{l,1}^2},
\end{align}
for all $(n,L,M)$. 
Thus,  the deterministic predictions under optimal power allocation satisfy 
\[
\Delta^\mathrm{dist}_L   \le \Delta^\mathrm{corr}_L.
\]
\end{lemma}

 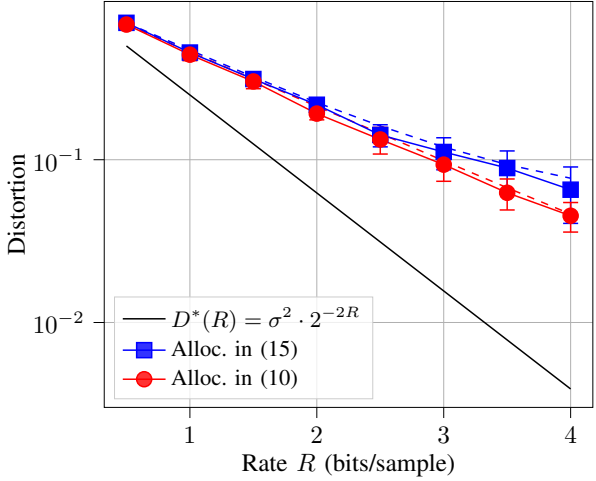
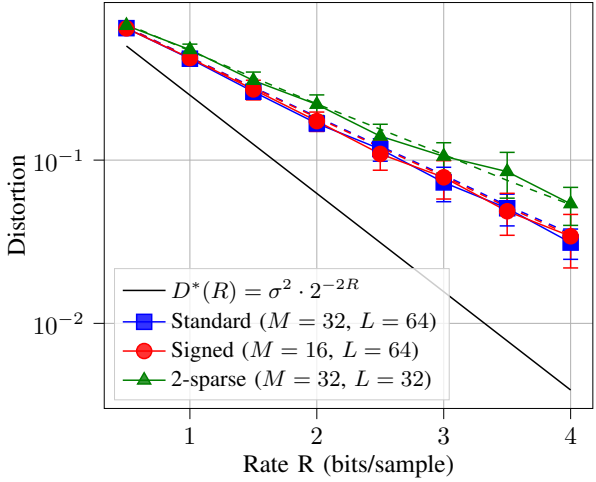
\begin{figure}[t]
 \vspace{0.1in}
   \centering
  \begin{subfigure}[t]{\linewidth}
 \resizebox{0.9\columnwidth}{!}{
\begin{tikzpicture}

\definecolor{darkgray176}{RGB}{176,176,176}
\definecolor{lightgray204}{RGB}{204,204,204}

\begin{axis}[
legend cell align={left},
legend style={fill opacity=0.8, draw opacity=1, text opacity=1, draw=lightgray204, at={(0.02,0.02)}, anchor=south west, font=\small},
log basis y={10},
tick align=outside,
tick pos=left,
x grid style={darkgray176},
xlabel={Rate  $R$  (bits/sample)},
xmajorgrids,
xmin=0.325, xmax=4.175,
xminorgrids,
xtick style={color=black},
y grid style={darkgray176},
ylabel={Distortion},
ymajorgrids,
ymin=0.00300887727121396, ymax=0.938051787301897,
yminorgrids,
ymode=log,
ytick style={color=black},
ytick={0.0001,0.001,0.01,0.1,1,10},
yticklabels={
  \(\displaystyle {10^{-4}}\),
  \(\displaystyle {10^{-3}}\),
  \(\displaystyle {10^{-2}}\),
  \(\displaystyle {10^{-1}}\),
  \(\displaystyle {10^{0}}\),
  \(\displaystyle {10^{1}}\)
}
]
\path [draw=blue, semithick]
(axis cs:0.5,0.664364792482455)
--(axis cs:0.5,0.722555571720784);

\path [draw=blue, semithick]
(axis cs:1,0.418977558745072)
--(axis cs:1,0.491993591286258);

\path [draw=blue, semithick]
(axis cs:1.5,0.285474180449077)
--(axis cs:1.5,0.340951979390782);

\path [draw=blue, semithick]
(axis cs:2,0.197326850112581)
--(axis cs:2,0.237944262704054);

\path [draw=blue, semithick]
(axis cs:2.5,0.119980040591095)
--(axis cs:2.5,0.164035065644652);

\path [draw=blue, semithick]
(axis cs:3,0.0867367256129824)
--(axis cs:3,0.136440320424456);

\path [draw=blue, semithick]
(axis cs:3.5,0.0645751201485431)
--(axis cs:3.5,0.113242456584482);

\path [draw=blue, semithick]
(axis cs:4,0.0406346875099136)
--(axis cs:4,0.0902738003591816);

\addplot [semithick, blue, mark=-, mark size=3, mark options={solid}, only marks, forget plot]
table {%
0.5 0.664364792482455
1 0.418977558745072
1.5 0.285474180449077
2 0.197326850112581
2.5 0.119980040591095
3 0.0867367256129824
3.5 0.0645751201485431
4 0.0406346875099136
};
\addplot [semithick, blue, mark=-, mark size=3, mark options={solid}, only marks, forget plot]
table {%
0.5 0.722555571720784
1 0.491993591286258
1.5 0.340951979390782
2 0.237944262704054
2.5 0.164035065644652
3 0.136440320424456
3.5 0.113242456584482
4 0.0902738003591816
};
\addplot [semithick, blue, dashed, forget plot]
table {%
0.5 0.69314756774064
1 0.471815455364515
1.5 0.322914205401513
2 0.224537745988502
2.5 0.161210776233737
3 0.119646811384774
3.5 0.0933587640103745
4 0.0769797271103068
};
\path [draw=red, semithick]
(axis cs:0.5,0.644130404039382)
--(axis cs:0.5,0.709492421809836);

\path [draw=red, semithick]
(axis cs:1,0.405956679966402)
--(axis cs:1,0.478115887052486);

\path [draw=red, semithick]
(axis cs:1.5,0.273069949440136)
--(axis cs:1.5,0.332116983784496);

\path [draw=red, semithick]
(axis cs:2,0.175811451372258)
--(axis cs:2,0.208295240915256);

\path [draw=red, semithick]
(axis cs:2.5,0.108374518763108)
--(axis cs:2.5,0.157976227636161);

\path [draw=red, semithick]
(axis cs:3,0.0737206842072956)
--(axis cs:3,0.112871883671246);

\path [draw=red, semithick]
(axis cs:3.5,0.0491299915469181)
--(axis cs:3.5,0.076174448552666);

\path [draw=red, semithick]
(axis cs:4,0.0359342981964722)
--(axis cs:4,0.0545875728484759);

\addplot [semithick, red, mark=-, mark size=3, mark options={solid}, only marks, forget plot]
table {%
0.5 0.644130404039382
1 0.405956679966402
1.5 0.273069949440136
2 0.175811451372258
2.5 0.108374518763108
3 0.0737206842072956
3.5 0.0491299915469181
4 0.0359342981964722
};
\addplot [semithick, red, mark=-, mark size=3, mark options={solid}, only marks, forget plot]
table {%
0.5 0.709492421809836
1 0.478115887052486
1.5 0.332116983784496
2 0.208295240915256
2.5 0.157976227636161
3 0.112871883671246
3.5 0.076174448552666
4 0.0545875728484759
};
\addplot [semithick, red, dashed, forget plot]
table {%
0.5 0.677678162507811
1 0.45994075472581
1.5 0.313081155488831
2 0.212815503223919
2.5 0.145082920416897
3 0.098486489344316
3.5 0.0672746313945108
4 0.046370981447011
};
\addplot [semithick, black]
table {%
0.5 0.5
1 0.25
1.5 0.125
2 0.0625
2.5 0.03125
3 0.015625
3.5 0.0078125
4 0.00390625
};
\addlegendentry{$D^*(R)=\sigma^2\cdot 2^{-2R}$}
\addplot [semithick, blue, mark=square*, mark size=3, mark options={solid}]
table {%
0.5 0.69346018210162
1 0.455485575015665
1.5 0.31321307991993
2 0.217635556408318
2.5 0.142007553117873
3 0.111588523018719
3.5 0.0889087883665125
4 0.0654542439345476
};
\addlegendentry{Alloc. in \eqref{eq:asymp_PA}}
\addplot [semithick, red, mark=*, mark size=3, mark options={solid}]
table {%
0.5 0.676811412924609
1 0.442036283509444
1.5 0.302593466612316
2 0.192053346143757
2.5 0.133175373199635
3 0.093296283939271
3.5 0.0626522200497921
4 0.0452609355224741
};
\addlegendentry{Alloc. in \eqref{eq:c_corr}}
\end{axis}

\end{tikzpicture}}
    \caption{Effect of power allocation (standard SPARC, $M=16$, $L=100$).}
    \label{subfig:exp1}
  \end{subfigure}  

  \vspace{0.15in}
  \begin{subfigure}[t]{\linewidth}
\resizebox{0.9\columnwidth}{!}{
\begin{tikzpicture}

\definecolor{darkgray176}{RGB}{176,176,176}
\definecolor{green01270}{RGB}{0,127,0}
\definecolor{lightgray204}{RGB}{204,204,204}

\begin{axis}[
legend cell align={left},
legend style={fill opacity=0.8, draw opacity=1, text opacity=1, draw=lightgray204, at={(0.02,0.02)}, anchor=south west, font=\small},
log basis y={10},
tick align=outside,
tick pos=left,
x grid style={darkgray176},
xlabel={Rate  R  (bits/sample)},
xmajorgrids,
xmin=0.325, xmax=4.175,
xminorgrids,
xtick style={color=black},
y grid style={darkgray176},
ylabel={Distortion},
ymajorgrids,
ymin=0.00301117006425995, ymax=0.923165986728302,
yminorgrids,
ymode=log,
ytick style={color=black},
ytick={0.0001,0.001,0.01,0.1,1,10},
yticklabels={
  \(\displaystyle {10^{-4}}\),
  \(\displaystyle {10^{-3}}\),
  \(\displaystyle {10^{-2}}\),
  \(\displaystyle {10^{-1}}\),
  \(\displaystyle {10^{0}}\),
  \(\displaystyle {10^{1}}\)
}
]
\path [draw=blue, semithick]
(axis cs:0.5,0.608161303273335)
--(axis cs:0.5,0.678309900371406);

\path [draw=blue, semithick]
(axis cs:1,0.37997281345061)
--(axis cs:1,0.456139828589567);

\path [draw=blue, semithick]
(axis cs:1.5,0.236479054814252)
--(axis cs:1.5,0.288316155548998);

\path [draw=blue, semithick]
(axis cs:2,0.150109805164446)
--(axis cs:2,0.185253902916554);

\path [draw=blue, semithick]
(axis cs:2.5,0.0984420342364534)
--(axis cs:2.5,0.135251590131867);

\path [draw=blue, semithick]
(axis cs:3,0.0556187735261504)
--(axis cs:3,0.0903741529142878);

\path [draw=blue, semithick]
(axis cs:3.5,0.0395403400073385)
--(axis cs:3.5,0.0618748433902687);

\path [draw=blue, semithick]
(axis cs:4,0.0246710670843672)
--(axis cs:4,0.0378172845804423);

\addplot [semithick, blue, mark=-, mark size=3, mark options={solid}, only marks, forget plot]
table {%
0.5 0.608161303273335
1 0.37997281345061
1.5 0.236479054814252
2 0.150109805164446
2.5 0.0984420342364534
3 0.0556187735261504
3.5 0.0395403400073385
4 0.0246710670843672
};
\addplot [semithick, blue, mark=-, mark size=3, mark options={solid}, only marks, forget plot]
table {%
0.5 0.678309900371406
1 0.456139828589567
1.5 0.288316155548998
2 0.185253902916554
2.5 0.135251590131867
3 0.0903741529142878
3.5 0.0618748433902687
4 0.0378172845804423
};
\addplot [semithick, blue, dashed, forget plot]
table {%
0.5 0.652511977265292
1 0.426978111203673
1.5 0.279626942457861
2 0.184357742232788
2.5 0.121637817102641
3 0.0810932735448402
3.5 0.0526632597226105
4 0.0354997761085252
};
\path [draw=red, semithick]
(axis cs:0.5,0.605205147152337)
--(axis cs:0.5,0.677842640017162);

\path [draw=red, semithick]
(axis cs:1,0.38813273057601)
--(axis cs:1,0.454683928256606);

\path [draw=red, semithick]
(axis cs:1.5,0.233791987322846)
--(axis cs:1.5,0.308470084204378);

\path [draw=red, semithick]
(axis cs:2,0.149896345882057)
--(axis cs:2,0.196917048522825);

\path [draw=red, semithick]
(axis cs:2.5,0.086707127165941)
--(axis cs:2.5,0.131468020572435);

\path [draw=red, semithick]
(axis cs:3,0.0577514961784878)
--(axis cs:3,0.0991939704318931);

\path [draw=red, semithick]
(axis cs:3.5,0.0346317888994292)
--(axis cs:3.5,0.0626041763807875);

\path [draw=red, semithick]
(axis cs:4,0.0218481107271137)
--(axis cs:4,0.0464733729655437);

\addplot [semithick, red, mark=-, mark size=3, mark options={solid}, only marks, forget plot]
table {%
0.5 0.605205147152337
1 0.38813273057601
1.5 0.233791987322846
2 0.149896345882057
2.5 0.086707127165941
3 0.0577514961784878
3.5 0.0346317888994292
4 0.0218481107271137
};
\addplot [semithick, red, mark=-, mark size=3, mark options={solid}, only marks, forget plot]
table {%
0.5 0.677842640017162
1 0.454683928256606
1.5 0.308470084204378
2 0.196917048522825
2.5 0.131468020572435
3 0.0991939704318931
3.5 0.0626041763807875
4 0.0464733729655437
};
\addplot [semithick, red, dashed, forget plot]
table {%
0.5 0.650489663167039
1 0.424353029140187
1.5 0.277065140803057
2 0.182127217592908
2.5 0.119812851703936
3 0.0796491032848848
3.5 0.0515695775398472
4 0.0346674712752343
};
\path [draw=green01270, semithick]
(axis cs:0.5,0.623185373278025)
--(axis cs:0.5,0.71163130459629);

\path [draw=green01270, semithick]
(axis cs:1,0.436015877669416)
--(axis cs:1,0.513395275836129);

\path [draw=green01270, semithick]
(axis cs:1.5,0.26846891000761)
--(axis cs:1.5,0.346239383279799);

\path [draw=green01270, semithick]
(axis cs:2,0.187708646461731)
--(axis cs:2,0.250875176238902);

\path [draw=green01270, semithick]
(axis cs:2.5,0.114554490514782)
--(axis cs:2.5,0.165824961216452);

\path [draw=green01270, semithick]
(axis cs:3,0.0832998932237309)
--(axis cs:3,0.127773774817913);

\path [draw=green01270, semithick]
(axis cs:3.5,0.0584609780374251)
--(axis cs:3.5,0.111537797828543);

\path [draw=green01270, semithick]
(axis cs:4,0.0398271961836078)
--(axis cs:4,0.0680850100496281);

\addplot [semithick, green01270, mark=-, mark size=3, mark options={solid}, only marks, forget plot]
table {%
0.5 0.623185373278025
1 0.436015877669416
1.5 0.26846891000761
2 0.187708646461731
2.5 0.114554490514782
3 0.0832998932237309
3.5 0.0584609780374251
4 0.0398271961836078
};
\addplot [semithick, green01270, mark=-, mark size=3, mark options={solid}, only marks, forget plot]
table {%
0.5 0.71163130459629
1 0.513395275836129
1.5 0.346239383279799
2 0.250875176238902
2.5 0.165824961216452
3 0.127773774817913
3.5 0.111537797828543
4 0.0680850100496281
};
\addplot [semithick, green01270, dashed, forget plot]
table {%
0.5 0.681513828953579
1 0.467187391117033
1.5 0.320843472945376
2 0.221018584507154
2.5 0.154662764992936
3 0.108245116395036
3.5 0.0751715631203415
4 0.0533210634392871
};
\addplot [semithick, black]
table {%
0.5 0.5
1 0.25
1.5 0.125
2 0.0625
2.5 0.03125
3 0.015625
3.5 0.0078125
4 0.00390625
};
\addlegendentry{$D^*(R)=\sigma^2\cdot 2^{-2R}$}
\addplot [semithick, blue, mark=square*, mark size=3, mark options={solid}]
table {%
0.5 0.64323560182237
1 0.418056321020089
1.5 0.262397605181625
2 0.1676818540405
2.5 0.11684681218416
3 0.0729964632202191
3.5 0.0507075916988036
4 0.0312441758324048
};
\addlegendentry{Standard ($M=32$, $L=64$)}
\addplot [semithick, red, mark=*, mark size=3, mark options={solid}]
table {%
0.5 0.641523893584749
1 0.421408329416308
1.5 0.271131035763612
2 0.173406697202441
2.5 0.109087573869188
3 0.0784727333051905
3.5 0.0486179826401083
4 0.0341607418463287
};
\addlegendentry{Signed ($M=16$, $L=64$)}
\addplot [semithick, green01270, mark=triangle*, mark size=3, mark options={solid}]
table {%
0.5 0.667408338937158
1 0.474705576752772
1.5 0.307354146643704
2 0.219291911350317
2.5 0.140189725865617
3 0.105536834020822
3.5 0.0849993879329839
4 0.0539561031166179
};
\addlegendentry{2-sparse ($M=32$, $L=32$)}
\end{axis}

\end{tikzpicture}}
    \caption{Comparison between SPARC variants, using correlation-based encoding and the allocation in \eqref{eq:c_corr}. }
    \label{subfig:exp2}
  \end{subfigure}
    \caption{Empirical distortion (solid lines with markers) vs. theoretical predictions (dashed) for SPARCs and variants, with correlation-based encoding.  }
    \label{fig:corr_based}
      \end{figure}

    \begin{figure}[t]
    \vspace{0.1in}
\resizebox{0.9\columnwidth}{!}{
\begin{tikzpicture}

\definecolor{darkgray176}{RGB}{176,176,176}
\definecolor{lightgray204}{RGB}{204,204,204}

\begin{axis}[
legend cell align={left},
legend style={fill opacity=0.8, draw opacity=1, text opacity=1, draw=lightgray204, at={(0.02,0.02)}, anchor=south west, font=\small},
log basis y={10},
tick align=outside,
tick pos=left,
x grid style={darkgray176},
xlabel={Rate  $R$  (bits/sample)},
xmajorgrids,
xmin=0.325, xmax=4.175,
xminorgrids,
xtick style={color=black},
y grid style={darkgray176},
ylabel={Distortion},
ymajorgrids,
ymin=0.00302460521616318, ymax=0.84077178770757,
yminorgrids,
ymode=log,
ytick style={color=black},
ytick={0.0001,0.001,0.01,0.1,1,10},
yticklabels={
  \(\displaystyle {10^{-4}}\),
  \(\displaystyle {10^{-3}}\),
  \(\displaystyle {10^{-2}}\),
  \(\displaystyle {10^{-1}}\),
  \(\displaystyle {10^{0}}\),
  \(\displaystyle {10^{1}}\)
}
]
\path [draw=blue, semithick]
(axis cs:0.5,0.550591710615486)
--(axis cs:0.5,0.651008700084007);

\path [draw=blue, semithick]
(axis cs:1,0.304609952719905)
--(axis cs:1,0.439421651148815);

\path [draw=blue, semithick]
(axis cs:1.5,0.188944606104866)
--(axis cs:1.5,0.306080800946761);

\path [draw=blue, semithick]
(axis cs:2,0.117856948970928)
--(axis cs:2,0.197899048653793);

\path [draw=blue, semithick]
(axis cs:2.5,0.073161300489757)
--(axis cs:2.5,0.138991301533527);

\path [draw=blue, semithick]
(axis cs:3,0.0357788413908567)
--(axis cs:3,0.0966133474603846);

\path [draw=blue, semithick]
(axis cs:3.5,0.0334075187354218)
--(axis cs:3.5,0.0561190094133862);

\path [draw=blue, semithick]
(axis cs:4,0.0199568893945601)
--(axis cs:4,0.0405302483377607);

\addplot [semithick, blue, mark=-, mark size=3, mark options={solid}, only marks, forget plot]
table {%
0.5 0.550591710615486
1 0.304609952719905
1.5 0.188944606104866
2 0.117856948970928
2.5 0.073161300489757
3 0.0357788413908567
3.5 0.0334075187354218
4 0.0199568893945601
};
\addplot [semithick, blue, mark=-, mark size=3, mark options={solid}, only marks, forget plot]
table {%
0.5 0.651008700084007
1 0.439421651148815
1.5 0.306080800946761
2 0.197899048653793
2.5 0.138991301533527
3 0.0966133474603846
3.5 0.0561190094133862
4 0.0405302483377607
};
\addplot [semithick, blue, dashed, forget plot]
table {%
0.5 0.622651367605595
1 0.393013621355629
1.5 0.252761937800674
2 0.162625796081859
2.5 0.107463401006791
3 0.068936368575066
3.5 0.0471294588275088
4 0.031832058710274
};
\path [draw=red, semithick]
(axis cs:0.5,0.528852334300504)
--(axis cs:0.5,0.638939836982302);

\path [draw=red, semithick]
(axis cs:1,0.301869696354981)
--(axis cs:1,0.416646077659191);

\path [draw=red, semithick]
(axis cs:1.5,0.19365253833886)
--(axis cs:1.5,0.286695936479547);

\path [draw=red, semithick]
(axis cs:2,0.108606996227367)
--(axis cs:2,0.182826695729159);

\path [draw=red, semithick]
(axis cs:2.5,0.0607530081875977)
--(axis cs:2.5,0.10869071055679);

\path [draw=red, semithick]
(axis cs:3,0.0407145866512687)
--(axis cs:3,0.0662643005777515);

\path [draw=red, semithick]
(axis cs:3.5,0.0212111052522857)
--(axis cs:3.5,0.0360613234336784);

\path [draw=red, semithick]
(axis cs:4,0.0150691534779206)
--(axis cs:4,0.0255416585488196);

\addplot [semithick, red, mark=-, mark size=3, mark options={solid}, only marks, forget plot]
table {%
0.5 0.528852334300504
1 0.301869696354981
1.5 0.19365253833886
2 0.108606996227367
2.5 0.0607530081875977
3 0.0407145866512687
3.5 0.0212111052522857
4 0.0150691534779206
};
\addplot [semithick, red, mark=-, mark size=3, mark options={solid}, only marks, forget plot]
table {%
0.5 0.638939836982302
1 0.416646077659191
1.5 0.286695936479547
2 0.182826695729159
2.5 0.10869071055679
3 0.0662643005777515
3.5 0.0360613234336784
4 0.0255416585488196
};
\addplot [semithick, red, dashed, forget plot]
table {%
0.5 0.617297939786157
1 0.378954107605365
1.5 0.233506406665893
2 0.141395681368193
2.5 0.0866223442918267
3 0.0505062535230465
3.5 0.0310635874051479
4 0.0188390578468854
};
\addplot [semithick, black]
table {%
0.5 0.5
1 0.25
1.5 0.125
2 0.0625
2.5 0.03125
3 0.015625
3.5 0.0078125
4 0.00390625
};
\addlegendentry{$D^*(R)=\sigma^2\cdot 2^{-2R}$}
\addplot [semithick, blue, mark=square*, mark size=3, mark options={solid}]
table {%
0.5 0.600800205349747
1 0.37201580193436
1.5 0.247512703525813
2 0.15787799881236
2.5 0.106076301011642
3 0.0661960944256206
3.5 0.044763264074404
4 0.0302435688661604
};
\addlegendentry{Corr-based, \eqref{eq:c_corr}}
\addplot [semithick, red, mark=*, mark size=3, mark options={solid}]
table {%
0.5 0.583896085641403
1 0.359257887007086
1.5 0.240174237409204
2 0.145716845978263
2.5 0.0847218593721937
3 0.0534894436145101
3.5 0.0286362143429821
4 0.0203054060133701
};
\addlegendentry{Dist-based, \eqref{eq:Delta_l_dist}}
\end{axis}

\end{tikzpicture}}
    \caption{Correlation- vs distance-based encoding, with optimal allocation for each (standard SPARC, $M=128$, $L=16$).}
    \label{fig:corr_v_dist}
 \end{figure}
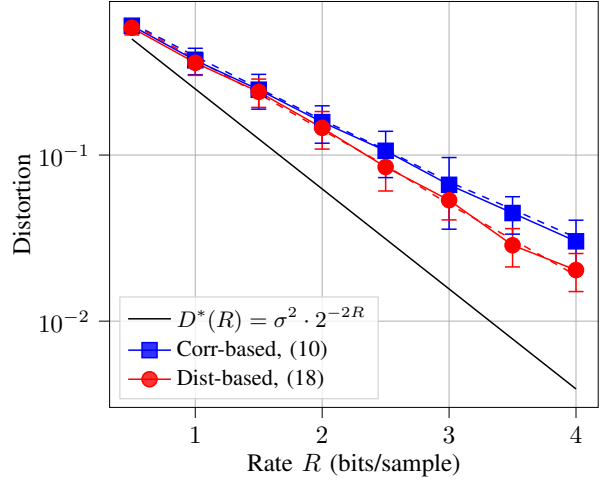

\section{Numerical Results}
\label{sec:experiments}

In Figures \ref{fig:corr_based} and \ref{fig:corr_v_dist}, we present three numerical experiments for an i.i.d.\ standard Gaussian source ($\sigma^2=1$). The matrix $A$ defining the code is also standard   Gaussian. Solid lines with markers show the average distortion over $10$ trials per rate point, with error bars indicating one standard deviation. Dashed lines show the theoretical distortion predictions; the solid black curve is the optimal distortion-rate function.

\emph{Effect of power allocation} ($M=16, L=100$):  Fig.~\ref{subfig:exp1} compares the allocation of \cite{RVGaussFeasible} (in \eqref{eq:asymp_PA}) with the optimal one in \eqref{eq:c_corr}, for the standard SPARC with correlation-based encoding. The  allocation in \eqref{eq:c_corr} achieves lower distortion across all rates. The gap is wider at higher rates,  due to $n = \frac{L \log M}{R} $ growing smaller,   hence closer to $\log M$ (see \eqref{eq:cl_heuristic}). Both theoretical predictions track the empirical curves closely.

\emph{SPARC variant comparison}:  Fig.~\ref{subfig:exp2} compares three SPARC variants, all using correlation-based encoding and the optimal allocation in \eqref{eq:c_corr}: the standard SPARC ($M=32$, $L=64$), the signed SPARC ($M=16$, $L=64$), and the 2-sparse SPARC ($M=32$, $L=32$). Thus, at each rate, the block length $n$ is the same for the standard and signed SPARCs, and about $10\%$ smaller for the 2-sparse SPARC.

The standard and signed SPARCs have very similar distortion-rate tradeoffs, suggesting that using  signed coefficients we can halve the size of the design matrix, with essentially no performance penalty.  The similar performance of standard and signed SPARCs can be understood by comparing \eqref{eq:w_sparc_asymp} and \eqref{eq:w_sparc_pm_asymp}: a signed SPARC with $M/2$ columns per section has $\omega_{l,\pm} = \psi(M) + \mathcal O(1/\log M)$, which is the width of a standard SPARC with $M$ columns per section to the same order.

In contrast, the 2-sparse SPARC (green) significantly underperforms the standard SPARC. Moreover, we found that the performance gap does not significantly improve  even with ($M=32,L=64$) for the 2-SPARC, which doubles its original block length. 
This can again be understood by comparing the asymptotic expansions for their Gaussian widths in \eqref{eq:w_sparc_asymp} and \eqref{eq:w_2sparc_asymp}: matching the number of codewords per section by using $M$ columns for the 2-SPARC and $\binom{M}{2}$ for the standard SPARC, both widths agree to leading order, since $\sqrt{2 \log \binom M2} = 2\sqrt{\log M} + \mathcal O(1/\sqrt{\log M})$, but the width of the 2-SPARC is smaller at the next order by
\begin{align*}
\omega_{l,1}\big(\tbinom{M}{2}\big) - \omega_{l,2}(M) = \frac{1}{4}\, \frac{\log \log M}{\sqrt{\log M}} + \mathcal O\!\left( \frac{1}{\sqrt{\log M}}\right).
\end{align*}

\emph{Correlation- vs distance-based encoding}:  Fig.~\ref{fig:corr_v_dist} compares correlation- and distance-based encoding, with each encoder given its optimal power allocation (\eqref{eq:c_corr} and via \eqref{eq:Delta_l_dist}, respectively).  Distance-based encoding achieves noticeably lower distortion at rates above $R \approx 1.5$ bits/sample. We note that the difference in performance between correlation- and distance-based encoders is much more pronounced when $L$ is small compared to $M$.  Indeed, in the extreme case of $L=1$, distance-based encoding is optimal. As $L$ is increased with $M$ kept constant, the gap between the two encoders shrinks, indicating that the sub-optimality of the greedy procedure becomes more significant than the difference between correlation- vs. distance-based selection rules within each section.

\section{Discussion}
\label{sec:discussion}

We studied additive orthogonal regression codes under two types of greedy encoding.
For any code in this family, our results specify an optimal choice of coefficients (based on the finite length code parameters and the type of encoder), and a deterministic prediction for the distortion with tight concentration guarantees. 
This significantly generalizes the result in \cite{RVGaussFeasible},  obtained for standard SPARCs with a specific power allocation under correlation-based encoding.  

Our results showed that with correlation-based encoding, the distortion-rate tradeoff of signed-SPARCs is nearly identical to that of standard SPARCs, allowing a 50\% reduction in the size of the design matrix. By contrast, the $2$-sparse SPARC was observed to have  higher distortion across all rates.  A key reason is that in the $2$-sparse SPARC, the two nonzero coefficients in each stage $l$ are both set equal to $\frac{c_l}{\sqrt 2}$. Allowing unequal coefficients could potentially improve the distortion-rate tradeoff.  Our results highlight the crucial role of the Gaussian width  \eqref{eq:omega_l_def} in determining the performance with correlation-based encoding, and an analogous quantity  \eqref{eq:gamma_l_def} for distance-based encoding. Using these to investigate the distortion-rate tradeoff for other examples of additive orthogonal regression codes is an interesting direction for future work.

\iflongpaper
\section*{Acknowledgement}
We used AI tools (Claude) to write code for the numerical experiments, to obtain the asymptotic expansions in Proposition \ref{prop:Gaussian_width}, and for final proofreading.
\fi


\bibliography{ITW2026_arxiv.bbl}


\iflongpaper



\appendices

\section{Gaussian concentration and iterated systems}


A  centered (i.e., zero-mean) random variable $X$ is said to be sub-Gaussian with variance proxy $\nu \in \bbR_+$ if 
\[
\log \E[ e^{ \lambda X}] \le  \frac{ \lambda^2\nu}{2} , \quad \text{for all $\lambda \in \bbR$}. 
\]

\begin{lemma}[\!\!{\cite[Section~2.3]{boucheron2013concentration}}]\label{lem:sub-Gaussian} If $X$ is sub-Gaussian with variance proxy $\nu$, then $\var(X)\le \nu$ and 
\[
\P( X - \E[X] \ge  t) \le \exp\left\{  - \frac{t^2}{2 \nu}\right\} , \quad \text{for all 
$t \in \bbR_+$}
\]
\end{lemma}

\begin{lemma}[Gaussian Concentration~{\cite[Theorem~5.5]{boucheron2013concentration}}]\label{lem:Gaussian_concentration}
Let $Z \sim \normal(0 , \Id_d)$ and let $f \colon \bbR^d \to \bbR$ be $c$-Lipschitz (with respect to the Euclidean norm). Then $f(Z) - \E[f(Z)]$ is sub-Gaussian with variance proxy $c^2$.
\end{lemma}

The following lemma shows that Lipschitz iterated systems driven by Gaussian noise yield sub-Gaussian outputs.

\begin{lemma}\label{lem:iter_con}
Let $Z_1,\dots, Z_L$ be i.i.d.\ $\normal(0, \Id_d)$ and let $X_1, \dots, X_L \in \bbR$ be generated by 
\begin{align*}
X_k &= f_k(X_{k-1}, Z_k),  \quad k=1, \dots, L 
\end{align*}
where $X_0= x_0\in \bbR$ is a deterministic  initialization. Suppose each $f_k\colon \bbR \times \bbR^{d} \to \bbR$ is $1$-Lipschitz in its first argument and $c_k$-Lipschitz in its second argument. Then, $X_l -\E[X_l]$ is sub-Gaussian with variance proxy $\sum_{k=1}^l c^2_k$.  Moreover $\E[X_l]$ is a $1$-Lipschitz function of $x_0$.
\end{lemma}
\begin{IEEEproof} It suffices to prove the result for $l=L$. By recursively expanding the definition of $X_L$, we can write 
\begin{align*}
X_L
&= f_L(X_{L-1}, Z_L) \\
&= f_L(f_{L-1}(X_{L-2}, Z_{L-1}), Z_L)  \\
&= \cdots =: F(Z_1,\dots,Z_L),
\end{align*}
for some  $F\colon (\mathbb{R}^d)^L \to \mathbb{R}$. We show that $F$ is Lipschitz.
 
Fix $k \in [L]$ and let $ (z_1,\dots,z_L)$ and $(z_1,\dots,z_k',\dots,z_L)$ differ only at index $k$. Let $(x_1, \dots, x_L)$ and $(x_1, \dots, x_{k-1},    x_k' ,\dots x'_L)$ be the corresponding sequences. At index $k$, we have
\[
|x_k -x_k'| = |f_k(x_{k-1}, z_k) - f_k(x_{k-1}, z_k')|
\le c_k \|z_k - z_k'\|.
\]
For $l > k$, using that $f_l$ is $1$-Lipschitz in its first argument,
\[
|x_l - x_l'| = |f_l(x_{l-1}, z_l) - f_l(x_{l-1}', z_l)| \le |x_{l-1} - x_{l-1}'|.
\]
Hence, $|x_L - x_L'| \le c_k \|z_k - z_k'\|$, which verifies 
\[
|F(z_1, \dots, z'_k, \dots, z_L) - F(z_1, \dots, z_L)|
\le c_k \|z_k - z_k'\|.
\]

For general $z= (z_1,\dots, z_L)$ and $z' = (z_1', \dots, z_L')$ define $z^{(k)} \coloneqq (z_1',\dots,z_k',z_{k+1},\dots,z_L)$. Then, 
\begin{align*}
 |F(z) - F(z')|
&\le \sum_{k=1}^L |F(z^{(k)}) - F(z^{(k-1)})| \\
&\le \sum_{k=1}^L c_k \|z_k - z_k'\| \\
&\le \Big(\sum_{k=1}^L c_k^2\Big)^{1/2}
\Big(\sum_{k=1}^L \|z_k - z_k'\|^2\Big)^{1/2}.
\end{align*}
Hence $F$ is Lipschitz with constant $\big(\sum_{k=1}^L c_k^2\big)^{1/2}$.  By Gaussian concentration (Lemma~\ref{lem:Gaussian_concentration}), $F(Z_1,\dots,Z_L)$ is sub-Gaussian with variance proxy $\sum_{k=1}^L c_k^2$. 

Finally, for fixed $z_1, \dots, z_L$, we see that $X_L$ is a $1$-Lipschitz function of $x_0$. Thus, its expectation is also  $1$-Lipschitz.
\end{IEEEproof}

\section{The mean of a Gaussian order statistic}

\begin{lemma}\label{lem:order_stat_mean}
Let $Z_1, \dots, Z_M$ be i.i.d.\ $\normal(0,1)$ with order statistics $Z_{(1)} \le \cdots \le Z_{(M)}$, let $\psi(\cdot)$ be as in Proposition~\ref{prop:Gaussian_width}, and set $\rho_M \coloneqq (\log \log M)^2 (\log M)^{-3/2}$. Then, for each $i \ge 1$, as $M \to \infty$,
\begin{align}
\E\big[ Z_{(M-i+1)}\big] & = \psi(M) - \frac{H_{i-1}}{\sqrt{2 \log M}} + \mathcal O(\rho_M), \label{eq:lem_topi} \\
\E\big[ \textstyle \max_{1 \le j \le M} \abs{Z_j} \big] & = \psi(2M) + \mathcal O(\rho_M), \label{eq:lem_absmax}
\end{align}
where \eqref{eq:lem_topi} holds with implied constant depending only on $i$.
\end{lemma}

For $i = 1$, \eqref{eq:lem_topi} is the classical expansion for the mean of a Gaussian maximum \cite{hall1979rate,GASULL2015376}. The general case is presumably also standard, but as we have not located a reference that states it in this form we give a short proof. 

\begin{IEEEproof} Let $\bar\Phi(x) \coloneqq \int_x^\infty (2 \pi)^{-1/2} e^{ -z^2/2} \, dz$ denote the standard Gaussian tail and $q \coloneqq \bar\Phi^{-1}$ its inverse. A standard representation of order statistics \cite[Sec.~2.1]{DavidNagaraja} yields
\[
Z_{(M-i+1)} \overset{d}{=} q(U_{(i)}), \quad U_{(i)} \sim \mathrm{Beta}(i, M-i+1).
\]

Set $W_i \coloneqq \log (M U_{(i)})$ and $m = \log M$.  Standard properties of the Beta distribution give
\[
 \E\big[ W_i \big]  = H_{i-1} -  \gamma +  \mathcal{O}\left(1/M \right) , \qquad  \E\big[W_i^2\big] = \mathcal{O}(1), \]
as well as the tail bound \[
\P\left[ \left|  W_i  \right| >   m/2   \right]    = \mathcal{O}\left(M^{-1/2} \right),\] where the implied constants may depend on $i$. 

The Gaussian quantile has the expansion (see, e.g., \cite{GASULL2015376}) 
\begin{align}
q(e^{-v} ) = \sqrt{2 v} - \frac{\log(4\pi v)}{2\sqrt{2v}} + \mathcal O\!\left( \frac{(\log v)^2}{v^{3/2}}\right), \notag   
\end{align}
as $v \to \infty$. Expanding this expression at $m$, uniformly for $|w|\le m/2$, gives
\begin{align}
 q\left( e^{ - m +  w } \right)  
& = \psi(M) - \frac{\gamma + w}{\sqrt{2 m}}  + \mathcal O\!\left( \frac{w^2 + (\log m)^2}{m^{3/2}}\right). 
\label{eq:app_taylor}
\end{align}
Since $U_{(i)} = e^{- m + W_i}$, we can apply \eqref{eq:app_taylor} with $w = W_i$ whenever $|W_i| \le m/2$. On the complementary event, Cauchy--Schwarz and
\[
\E \big[ Z^2_{(M-i+1)}\big] \le \E \left[ \max_j \abs{Z_j}^2 \right]= \mathcal{O}(m) 
\]
show that its contribution is  $\mathcal O(M^{-1/4}\sqrt{m})= o(\rho_M)$. Since $\E[ W_i^2] = \mathcal{O}(1)$, taking expectations in \eqref{eq:app_taylor} therefore gives 
\begin{align*}
\E\big[ Z_{(M-i+1)} \big]
&= \psi(M) - \frac{\gamma + \E[W_i]}{\sqrt{2 m}} + \mathcal O\!\left(\rho_M \right)\\
& = \psi(M) - \frac{H_{i-1}}{\sqrt{2 \log M}} + \mathcal O(\rho_M),
\end{align*}
proving \eqref{eq:lem_topi}.

For the absolute maximum, note that $\P(\abs{Z_1} > t) = 2 \bar\Phi(t)$, which implies that $\max_{j} \abs{Z_j} \overset{d}{=} Q(U_{(1)}/2)$. Since $U_{(1)}/2 = e^{-\log (2M) + W_1}$ and $\gamma + \E[W_1] = \mathcal O(1/M)$, the same quantile expansion, now centered at $\log(2M)$, yields  \eqref{eq:lem_absmax}. 
\end{IEEEproof}

\section{Proofs of main results} \label{app:proofs_main_results}

\subsection{Proof of Theorem~\ref{thm:concentration}}

We address both distance-based and correlation-based decoding rules. Throughout we treat the source $x \in \bbR^n$ as fixed.  Recall that the residual process is
\[
r_l = x - A \sum_{k=1}^l b_k.
\]
where each codeword $b_l$ is the solution to one of the following optimization problems:
\begin{alignat*}{3}
&\min_{b \in \cB_l} \| r_{l-1} - A b\|&\quad &\text{(distance-based)}\\
&\max_{b \in \cB_l} \langle A^\top r_{l-1} ,  b \rangle &\quad &\text{(correlation-based)}
\end{alignat*}

The key idea is to rotate coordinates at each stage so that the residual $r_{l-1}$  always points along the fixed direction $e_1$. This allows us to express the residual norm at the current stage as a function of the residual norm at the previous stage and a fresh, independent copy of the design matrix. We then apply the concentration result for iterated-Lipschitz systems from Lemma~\ref{lem:iter_con}.

\begin{lemma}\label{lem:Al}
Let $P_l$ be the orthogonal projection onto the span of $\cB_l$.  Let $Q_{l-1} \in \mathbb{O}(n)$ be measurable with respect to  $r_{l-1}$ and satisfy $ Q^\top_{l-1} e_1  = r_{l-1}/\|r_{l-1}\|$. Let $A'_1, \dots, A'_L$ be i.i.d.\ copies of $A$ and define
\[
A_l = Q_{l-1} A P_l + A'_l (\Id - P_l)
\]
Then, the  matrices $A_1, \dots , A_L$  are i.i.d.\ with $A_l\overset{d}{=} A$. 
\end{lemma}
\begin{IEEEproof}
Let $\cF_{l-1} = \sigma(AP_1,\dots,AP_{l-1})$. Then $r_{l-1}$ and $Q_{l-1}$ are $\cF_{l-1}$-measurable. By orthogonality of $\{P_k\}$, $AP_l$ is independent of $\cF_{l-1}$. By left-rotational invariance of $AP_l$, $Q_{l-1} A P_l$ has the same distribution as $A P_l$ and is independent of $\cF_{l-1}$. Since $A'_l(\Id - P_l)$ is independent of everything else and has the same distribution as $A(\Id - P_l)$, the claim follows.
\end{IEEEproof}

Using Lemma~\ref{lem:Al},  the residual norm admits the recursion
\begin{align*}
 \|r_l\|  & = \| r_{l-1} - A b_l\| = \| Q_{l-1}( r_{l-1} - A b_l) \|\\
 & = \big\|  \|r_{l-1}\| e_1 -  A_l b_l\big \|.
\end{align*}
where $A_{l}$ has i.i.d.\ $\normal(0,1)$ entries and is independent of $(r_1, \dots, r_{l-1})$.

For distance-based decoding, 
\begin{align}
\|r_l\| = \min_{b \in \cB_l} \big\|  \|r_{l-1}\| e_1 -  A_l b\big \|. \label{eq:rl_decomp_dist}
\end{align}
This mapping is $1$-Lipschitz in $\|r_{l-1}\|$ and $c_l$-Lipschitz in $A_l$ (with respect to the Frobenius norm). Hence, by Lemma~\ref{lem:iter_con}, $\|r_l\|$ is sub-Gaussian with variance proxy $\sum_{k=1}^l c_k^2$.

For correlation-based decoding, let $a_{li}$  denote the $i$-th row of $A_l$. Then 
\[
b_l = \argmax_{b \in \cB_l} \langle a_{l1}, b \rangle
\]
depends only on the first row. 
Since $A_l$ has i.i.d.\ $\normal(0,1)$ entries, its rows are independent, and the projection of $b_l$ onto the remaining rows $a_{l2}, \dots, a_{ln}$ is independent of $a_{l1}$ with i.i.d.\ standard Gaussian entries: 
\[
z_l = \frac{1}{c_l} \big( \langle a_{l2}, b_l \rangle, \dots, \langle a_{ln}, b_l \rangle \big) \sim \normal(0, \Id_{n-1})
\]
Then
\begin{align}
\|r_{l}\| 
&= \sqrt{ (\|r_{l-1}\|- \langle a_{l1}, b_l \rangle)^2 + \sum_{i=2}^n \langle a_{li}, b_l \rangle^2 } \notag \\
& = \sqrt{
\big(\|r_{l-1}\|- \max_{b \in \cB_l} \langle a_{l1}, b \rangle \big)^2
+ c_l^2 \|z_l\|^2 }, \label{eq:rl_decomp_corr}
\end{align}
where $\|r_{l-1}\|$, $a_{l1}$, and $z_l$ are mutually independent.  This mapping is $1$-Lipschitz in $\|r_{l-1}\|$ and $c_l$-Lipschitz in Gaussian variables $(a_{l1}, z_l)$. The result follows from Lemma~\ref{lem:iter_con}.

\subsection{Proof of Theorem~\ref{thm:D_corr}}

Throughout this proof we condition on the event $\|x\| = \sqrt{n} \sigma$. Squaring both sides of \eqref{eq:rl_decomp_corr} and then taking the expectation (recalling that $r_{l-1}$ and $a_{l1}$ are independent) yields
\begin{align}
\E[ \|r_{l}\|^2 ]
& =  \E[\|r_{l-1}\|^2] + \E\big[  \big(\max_{b \in \cB_l}\langle a_{l1}, b \rangle  \big)^2 \big] + c_l^2 (n-1) \notag \\
& \quad 
- 2 \E[ \|r_{l-1}\|] \E\big[  \max_{b \in \cB_l}\langle a_{l1}, b \rangle \big], \label{eq:rl2_decomp}
\end{align}
where $a_{l1} \sim \normal(0, \Id_N)$.  From the definition of the Gaussian width in \eqref{eq:omega_l_def},
\begin{align}
\E[  \max_{b \in \cB_l} \langle a_{l1}, b \rangle ]  &= c_l \omega_l, \label{eq:Ecorr} 
\end{align}
Using $D_l = \frac{1}{n} \E[\|r_l\|]^2 = \frac{1}{n}  \E[\|r_l\|^2] -\frac{1}{n} \var(\|r_l\|)$ then leads to
\begin{align*}
D_l & = \Big(  \sqrt{D_{l-1}}  -  \frac{ c_l \omega_l}{\sqrt{n}} \Big)^2+ c_l^2 +  \frac{ \eta_l}{n}
\end{align*}
where
\begin{align*}
\eta_l & =\var(\|r_{l-1}\|) -\var(\|r_l\|) + \var\big(\max_{b \in \cB_l}\langle a_{l1}, b \rangle  \big) - c_l^2.
\end{align*}
By Theorem~\ref{thm:concentration},  $\var(\|r_l\|) \le \sum_{k=1}^l c_k^2$. Also, since $\max_{b \in \cB_l}\langle a_{l1}, b \rangle$ is a $c_l$-Lipschitz function of $a_{l1}$, Gaussian concentration (Lemma~\ref{lem:sub-Gaussian} and Lemma~\ref{lem:Gaussian_concentration}) yields
\begin{align}
\var\big(\max_{b \in \cB_l}\langle a_{l1}, b \rangle\big)  \le c_l^2.  \label{eq:corr_var}
\end{align}
Thus $|\eta_l| \le 2 \sum_{k=1}^l c_k^2$. \IEEEQED

\subsection{Power allocation under correlation-based coding} 
\label{app:power_alloc_corr}

Expanding the square in \eqref{eq:Delta_l_corr} yields
\begin{align*}
\Delta_l = \Delta_{l-1} - \frac{2\sqrt{\Delta_{l-1}}\,\omega_l}{\sqrt{n}}\, c_l + \left(1 + \frac{\omega_l^2}{n}\right) c_l^2.
\end{align*}
The right-hand side is a quadratic in $c_l$ with positive leading coefficient, so it is minimized at
\begin{align}
c_l = \frac{\sqrt{n}\,\omega_l}{n+\omega_l^2}\sqrt{\Delta_{l-1}}, \label{eq:app_cl_vertex}
\end{align}
which gives the minimal value
\begin{align*}
\Delta_l &= \Delta_{l-1} - \frac{n \omega_l^2}{(n+\omega_l^2)^2}\Delta_{l-1} \cdot \frac{n+\omega_l^2}{n} \\
& = \Delta_{l-1}\left(1 - \frac{\omega_l^2}{n+\omega_l^2}\right) = \frac{n}{n+\omega_l^2}\,\Delta_{l-1}.
\end{align*}
Iterating from $\Delta_0 = \sigma^2$ gives \eqref{eq:Delta_corr_opt}. Substituting \eqref{eq:Delta_corr_opt} into \eqref{eq:app_cl_vertex} and squaring gives
\begin{align*}
c_l^2&  = \frac{n\,\omega_l^2}{(n+\omega_l^2)^2}\,\Delta_{l-1} = \frac{\sigma^2 \omega_l^2}{n+\omega_l^2}\cdot\frac{n}{n+\omega_l^2}\prod_{k=1}^{l-1}\frac{n}{n+\omega_k^2} \\
& = \frac{\sigma^2\omega_l^2}{n+\omega_l^2}\prod_{k=1}^{l}\frac{n}{n+\omega_k^2},
\end{align*}
which is \eqref{eq:c_corr}.

Finally, we verify that \eqref{eq:c_corr} satisfies the total power constraint \eqref{eq:ckUB}. Comparing \eqref{eq:c_corr} with $\Delta^\mathrm{corr}_{l-1} - \Delta^\mathrm{corr}_{l} = \frac{\omega_l^2}{n + \omega_l^2}\Delta^\mathrm{corr}_{l-1}$ gives
\begin{align*}
c_l^2 = \frac{n}{n+\omega_l^2}\big( \Delta^\mathrm{corr}_{l-1} - \Delta^\mathrm{corr}_{l}\big) \le \Delta^\mathrm{corr}_{l-1} - \Delta^\mathrm{corr}_{l},
\end{align*}
and summing the telescoping series yields
\[\sum_{l=1}^L c_l^2 \le \sigma^2 - \Delta^\mathrm{corr}_L \le \sigma^2.
\]

\subsection{Proof of Theorem~\ref{thm:D_corr_opt}}

Throughout this proof we condition on the event $\|x\| = \sqrt{n} \sigma$. We begin with the lower bound. By \eqref{eq:rl_decomp_corr}, convexity of the map $(x,y) \mapsto \sqrt{ x^2 + y^2}$, and Jensen's inequality,
\begin{align*}
\E[ \|r_{l}\| ]
& \ge  \sqrt{
\big(\E[ \|r_{l-1}\|] - \E\big[  \max_{b \in \cB_l} \langle a_{l1}, b \rangle \big] \big)^2
+ c_l^2 \E[ \|z_l\|]^2 }, 
\end{align*}
where $a_{l1} \sim \normal(0, \Id_N)$ and $z_l \sim \normal(0, \Id_{n-1})$. 
Since $\|z_l\|$ is a $1$-Lipschitz function of $z_l$, its variance is no greater than $1$ by the Gaussian concentration  (Lemma~\ref{lem:sub-Gaussian} and Lemma~\ref{lem:Gaussian_concentration}), and thus
\begin{align*}
  \E[\|z_l\|]^2 \ge  \E[\|z_l\|^2]  - 1=  n-2.
\end{align*}
Using this inequality along with  \eqref{eq:Ecorr}, we can write  
\begin{align*}
\E[ \|r_{l}\| ]^2
& \ge 
\big(\E[ \|r_{l-1}\|] - c_l \omega_l \big)^2
+ c_l^2 (n-2)\\
& \ge   \inf_{c \in \bbR_+} \{
\big(\E[ \|r_{l-1}\|] - c \omega_l \big)^2
+ c^2 (n-2)\}\\
& =  \Big(\frac{n - 2}{n -2 + \omega_l^2}  \Big) \E[ \|r_{l-1}\|]^2 
\end{align*}
Iterating this inequality and  recalling that $\|r_0\|= \sqrt{n} \sigma$ (by assumption) establishes  the  lower bound
\begin{align}
\E[ \|r_{l}\| ]^2
& \ge n \sigma^2 \prod_{k=1}^l   \frac{n-2}{n -2 + \omega_k^2}, \label{eq:rlLB}
\end{align}
which is valid for any power allocation. To express the bound in terms of $D_l$ and $\Delta^\mathrm{corr}_l$, we divide by $n$ and write
\begin{align}
D_l = \frac{1}{n}\E[\|r_{l}\|]^2
& \ge    \Delta_l^\mathrm{corr} (1- \eps_l) \label{eq:DlLB}
\end{align}
where 
\begin{align*}
\eps_{l} 
& \coloneqq  1-  \prod_{k=1}^{l} \frac{ (n-2 ) (n+ \omega_k^2)}{n (n-2+\omega_k^2)} \le  \frac{2}{n} \sum_{k=1}^{l} \frac{ \omega_k^2}{  n - 2 + \omega_k^2}.
\end{align*}

Next, we consider the upper bound. Combining \eqref{eq:rl2_decomp}, \eqref{eq:Ecorr}, and \eqref{eq:corr_var} gives:
\begin{align*}
\E[ \|r_{l}\|^2 ]& \le  \E[\|r_{l-1}\|^2] + c_l^2 (n+ \omega_l^2) - 2 c_l \omega_l \E[\|r_{l-1}\|].
\end{align*}
Applying the lower bound on $\E[\|r_{l-1}\|]$ in \eqref{eq:DlLB} along with the basic inequality $\sqrt{1 -x}  \ge 1- x$ for all $0 \le x \le 1$ yields
\begin{align}
\E[ \|r_{l}\|^2 ]&\ \le  \E[\|r_{l-1}\|^2] + c_l^2 (n+ \omega_l^2)  \notag \\
&\quad - 2 c_l \omega_l \sqrt{n \Delta^\mathrm{corr}_{l-1}} (1- \eps_{l-1})   . \notag
\end{align}
Under the power allocation in  \eqref{eq:c_corr}, this inequality becomes
\begin{align*}
\E[ \|r_{l}\|^2 ]
& \le  \E[\|r_{l-1}\|^2] -   \frac{n \omega_l^2}{n + \omega_l^2}  \Delta^\mathrm{corr}_{l-1} (1 -2  \eps_{l-1}).
\end{align*} 
Iterating this bound, and then using the identity $ \omega_l^2 \Delta^\mathrm{corr}_{l-1} = ( \Delta^\mathrm{corr}_{l-1} - \Delta^\mathrm{corr}_l) (n + \omega_l^2)$ along with $\|r_0\|^2 = n \sigma^2$ and  $\Delta_0^\mathrm{corr} = \sigma^2$ leads to 
\begin{align*}
\E[ \|r_{l}\|^2 ]  & \le \E[ \|r_0\|^2 ]  -   \sum_{k=1}^{l}   \frac{n \omega_k^2}{n + \omega_k^2}  \Delta^\mathrm{corr}_{k-1} (1 - 2 \eps_{k-1})\\
& =  n   \Delta^\mathrm{corr}_l + 2 n  \sum_{k=1}^{l}   (\Delta^\mathrm{corr}_{k-1} -  \Delta^\mathrm{corr}_{k}) \eps_{k-1}.
\end{align*}
Furthermore, since  $\Delta^\mathrm{corr}_{k-1} \ge  \Delta^\mathrm{corr}_{k}$ and $\eps_k$ is increasing, the summation satisfies
\begin{align*}
\sum_{k=1}^{l}     (\Delta^\mathrm{corr}_{k-1} -  \Delta^\mathrm{corr}_{k}) \eps_{k-1} \le ( \sigma^2 - \Delta^\mathrm{corr}_l) \eps_{l-1}
\end{align*}

To relate these bounds back to the deterministic prediction $D_l$ under the same power allocation, we can write  
\begin{align*}
D_l  & =  \frac{1}{n} \E[\|r_l\|]^2 \le  \frac{1}{n} \E[\|r_l\|^2] \le  \Delta^\mathrm{corr}_l + 2 ( \sigma^2 - \Delta^\mathrm{corr}_l) \eps_{l-1}.  
\end{align*}
This concludes the proof. \IEEEQED

\subsection{Proof of Proposition~\ref{prop:Gaussian_width}}

In a standard SPARC, the coefficient set $\cB_l$ has $M$ vectors, each with a single nonzero entry in a different location. Therefore, in \eqref{eq:omega_l_def} the set of random variables $\{ \langle z, b \rangle / \|b\| \}_{b \in \cB_l}$ are $M$ i.i.d.\ standard normals, so $\omega_{l,1} = \E[Z_{(M)}]$. The only difference for the signed SPARC is that for each $b \in \cB_l$, $-b$ is also in $\cB_l$, hence the maximum is over the absolute values of $M$ independent standard Gaussians.

For the  $K$-sparse SPARC, each coefficient vector in $\cB_l$ has $K$ nonzeros in section $l$, each with value $c_l/\sqrt K$, and $\cB_l$ ranges over all size-$K$ subsets of the $M$ coordinates. Since the $K$ nonzero coefficients are equal and positive, the sum $\sum_{i \in S} Z_i$ over a size-$K$ subset $S$ is maximized by taking the $K$ largest $Z_i$'s, so
\begin{align*}
\omega_{l,K} &= \frac{1}{\sqrt K}\,\E\Big[\max_{1 \le i_1 < \cdots < i_K \le M} (Z_{i_1} + \cdots + Z_{i_K})\Big] \\
& = \frac{1}{\sqrt K}\, \E\big[ Z_{(M)} + \cdots + Z_{(M-K+1)}\big].
\end{align*}
For each fixed $i \ge 1$, Lemma~\ref{lem:order_stat_mean} gives
\begin{equation}
\E\big[Z_{(M-i+1)}\big] = \psi(M) - \frac{H_{i-1}}{\sqrt{2 \log M}} + \mathcal O(\rho_M). \label{eq:app_orderstat}
\end{equation}
The case $i=1$, where $H_0 = 0$, is \eqref{eq:w_sparc_asymp}. Summing \eqref{eq:app_orderstat} over $i = 1, \dots, K$ and using $\sum_{i=1}^{K} H_{i-1} = K H_{K-1} - (K-1)$ gives
\begin{align*}
\omega_{l,K} &= \frac{1}{\sqrt K} \sum_{i=1}^K \E\big[Z_{(M-i+1)}\big] \\
&= \sqrt K\, \psi(M) - \frac{K H_{K-1} - K + 1}{\sqrt{2 K \log M}} + \mathcal O(\rho_M),
\end{align*}
which is \eqref{eq:w_2sparc_asymp}. The corresponding expansion for the signed SPARC in \eqref{eq:w_sparc_pm_asymp} follows directly from \eqref{eq:lem_absmax}. 
\IEEEQED

\subsection{Proof of Theorem~\ref{thm:D_dist}}
From  \eqref{eq:rl_decomp_dist}, we see that
\begin{align*}
\E\big[ \|r_l\| \mid \|r_{l-1}\| \big] & = g_l( \|r_{l-1}\|)
\end{align*}
where
\begin{align*}
g_l(u) & \coloneqq  \E\Big[\min_{b \in \cB_l} \big\|  u e_1 -  A_l b\big \|\Big]  = u \gamma_l\Big(n, \frac{c_l}{u} \Big)
\end{align*}
Since $g_l(\cdot)$ is $1$-Lipschitz continuous,
 \begin{align*}
\MoveEqLeft \abs{\E[\|r_l\|] - g_l(\E[\|r_{l-1}\|]) } \\
& = \abs{\E[g_l(\|r_{l-1}\|)] - g_l(\E[\|r_{l-1}\|]) }\\
& \le \E\big[ \abs{\|r_{l-1}\| - \E[\|r_{l-1}\|] } \big]\\
& \le \sqrt{ \var(\|r_{l-1}\|)}.
\end{align*}
Thus
\begin{align*}
\sqrt{D_l} & = \frac{1}{\sqrt{n}} \E[  g_l( \|r_{l-1}\|)] \\
& = \frac{1}{\sqrt{n}} g_l(\sqrt{n D_{l-1}}) + \frac{\eta_l}{\sqrt{n}}\\
& = \sqrt{D_{l-1}} \gamma_l\Big(n, \frac{c_l}{\sqrt{n D_{l-1}}}\Big) + \frac{\eta_l}{\sqrt{n}}
\end{align*}
where $|\eta_l| \le \sqrt{\var(\|r_{l-1} \|)} \le (\sum_{k=1}^{l-1} c_k^2)^{1/2}$. The last inequality follows from Lemma \ref{lem:iter_con}.

\subsection{Proof of Proposition~\ref{prop:gamma_expressions}}
Each expression follows by substituting the coefficient set $\cB_l$ of the corresponding SPARC variant into the definition of $\gamma_l(n,\lambda)$ in \eqref{eq:gamma_l_def}, and then taking the infimum over $\lambda$. For the standard SPARC, $\cB_l = \{c_l e_j\}_{j=1}^M$ with $\|b\|=c_l$ for every $b \in \cB_l$. Writing $z_j \coloneqq A e_j$, so that $z_1,\dots,z_M \sim_{\text{i.i.d.}} \normal(0,\Id_n)$, \eqref{eq:gamma_l_def} becomes $\gamma_l(n,\lambda) = \E[\min_{1\le j \le M} \|e_1 - \lambda z_j\|]$. For the signed SPARC, $\cB_l = \{\pm c_l e_j\}_{j=1}^M$, so the minimization in \eqref{eq:gamma_l_def} additionally ranges over the sign of each $z_j$, giving the $\min\{\|e_1-\lambda z_j\|, \|e_1+\lambda z_j\|\}$ term. For the $K$-sparse SPARC, $\cB_l = \big\{ \frac{c_l}{\sqrt K}\sum_{i \in S} e_i : S \subset [M], |S|=K \big\}$, so $b/\|b\|$ ranges over $\frac{1}{\sqrt K}\sum_{i\in S} e_i$ for size-$K$ subsets $S$, and writing $z_i \coloneqq Ae_i$ gives the stated expression for $\bar\gamma_{l,K}(n)$. \IEEEQED

\subsection{Proof of Lemma~\ref{lem:gammabar_UB}}

Recall from Proposition \ref{prop:gamma_expressions} that for the standard SPARC with parameters $(n,L,M)$, we have
\begin{align*}
\bar{\gamma}_{l,1}(n) &= \inf_{\lambda \in \bbR^+} \E\left[  \min_{ 1 \le j \le M} \| e_1 - \lambda z_j\| \right],
\end{align*}
where $z_1, \dots, z_M$ are i.i.d.\ $\normal(0, \Id_n)$.

Let $j^* = \argmax_{1 \le j \le M} \langle e_1, z_j \rangle $. Then, for every $\lambda \in \bbR^+$, 
\begin{align}
 & \min_{ 1 \le j \le M}  \| e_1 - \lambda z_j\|^2  \le   \| e_1 - \lambda z_{j^*}\|^2 \notag \\
 & = 1 - 2 \lambda  \langle e_1, z_{j^*} \rangle + \lambda^2  \langle e_1, z_{j^*} \rangle^2   + \lambda^2  \sum_{i=2}^n \langle e_i, z_{j^*} \rangle^2 . \label{eq:minjtojstar}
\end{align}

We now bound each term. By definition of the Gaussian width, $ \E[\langle e_1, z_{j^*} \rangle] = \omega_{l,1}$. 
Moreover, 
\[
\E[\langle e_1, z_{j^*} \rangle^2]
= \omega_{l,1}^2 + \var(\langle e_1, z_{j^*} \rangle)
\le \omega_{l,1}^2 + 1,
\]
where the variance bound follows from the Gaussian Poincar\'e inequality \cite[Theorem~3.21]{boucheron2013concentration} since
$\max_{1\le j\le M} \langle e_1, z_j\rangle$ is $1$-Lipschitz as a function of
$(\langle e_1,z_j\rangle)_{j=1}^M$.

Finally, since $j^*$ depends only on the first coordinate,
the variables $(\langle e_i, z_{j^*} \rangle)_{i=2}^n$ are i.i.d.\ $\normal(0,1)$, and hence
\[
\sum_{i=2}^n \E[\langle e_i, z_{j^*} \rangle^2] = n-1.
\]

Taking expectations in \eqref{eq:minjtojstar} and using Jensen's inequality,
\begin{align*}
& \E \left[ \min_{ 1 \le j \le M}  \| e_1 - \lambda z_j\| \right]^2  \le  \E \left[ \min_{ 1 \le j \le M}  \| e_1 - \lambda z_j\|^2 \right]\\
 & \quad    \le  1 - 2 \lambda \omega_{l,1} + \lambda^2 \omega_{l,1}^2 +  \lambda^2 n  
\end{align*}
Minimizing over $\lambda>0$ gives  $\lambda^* = \omega_{l,1}/( n + \omega_{l,1}^2)$, which leads to the stated inequality. \IEEEQED

\fi

\end{document}